\documentclass{article}
\usepackage[ansinew]{inputenc}
\usepackage{amsmath,amssymb,amsthm,amscd}
\usepackage{hyperref}
\usepackage{here, caption}
\usepackage{graphicx}
\DeclareMathOperator{\supp}{supp}
\DeclareMathOperator{\id}{id}

\newtheorem{thm}{Theorem}[section]
\newtheorem{prop}{Proposition}[section]
\newtheorem{lem}{Lemma}[section]

\newtheorem{defi}{Definition}[section]

\begin{document}

\title{Quantization of Maxwell's equations on curved backgrounds and general local covariance}
\author{Claudio Dappiaggi\\
  Dipartimento di Fisica Nucleare e Teorica,\\
  Universit\`a degli Studi di Pavia \& INFN Sezione di Pavia,\\
  Via Bassi 6, I-27100 Pavia, Italy\\
  \texttt{claudio.dappiaggi@pv.infn.it} \vspace{.5cm}\\
 Benjamin Lang\footnote{current address: Department of Mathematics, University of York, Heslington, York, UK. YO10 5DD}\\
  II. Institut f\"ur Theoretische Physik\\
  Universit\"at Hamburg\\
  Luruper Chaussee 149, D-22761 Hamburg, Deutschland\\
  \texttt{bl620@york.ac.uk}}
\date{\today}
\maketitle

\begin{abstract}
We develop a quantisation scheme for Maxwell's equations without source on an oriented 4-dimensional globally hyperbolic spacetime with at most finitely many connected components. The field strength tensor is the key dynamical object and it is not assumed a priori that it descends from a vector potential. It is shown that, in general, the associated field algebra can contain a non trivial centre and, on account of this, such a theory cannot be described within the framework of general local covariance unless further restrictive assumptions on the topology of the spacetime are made.
\end{abstract}

\section{Introduction}

Electromagnetic interactions played a key role in the history of physics since they are related to the first successful example of unification of two apparently different fields, the electric and the magnetic one, into a single body, the Faraday tensor $F$. The latter fulfills the so-called Maxwell's equations which, on a flat background, are automatically Poincar\'e invariant and they yield that $F$ can be described in terms of an auxiliary field, the vector potential $A$. Even though $F$ stays the basic observable of the theory, $A$ has the advantage of being apparently easier to handle since every Faraday tensor, solution of Maxwell's equations, can be reconstructed from a vector potential which solves both the wave equation and a second one, known as the Lorenz gauge. With the advent of quantum field theory, this interplay between $A$ and $F$ has been even more emphasized since the quantization scheme which is canonically employed still focuses on the vector potential and considers $F$, also known as the field strength, as a derived object, albeit it is the real observable. In view of the Aharonov-Bohm effect this latter assertion might be considered even erroneous since it is often stated that it is actually $A$ the true physical object. Yet, as noted for example in \cite[\S 2.6]{Sakurai}, in all idealized and real experiments of the Aharonov-Bohm kind the true observable is actually the flux of the magnetic field which is present inside an impenetrable region, typically a cylinder. Hence even this quantity can be expressed in terms of the components of the field strength tensor. The role of the vector potential becomes primary as soon as interactions are switched on, but, till we consider only a free Maxwell system, $F$ should contain all the physical information both at a classical and at a quantum level.

It is far from the scope of this paper to discuss the details of this procedure, but suffice to say that, on Minkowski background and in absence of sources, the result is pretty much satisfactory. Yet the situation starts to complicate itself as soon as it is assumed that the spacetime $M$ has a non trivial geometry. Although we shall provide more details in the main body of the paper, we can easily explain the source of all potential problems. The field strength tensor is best described as a two-form $F\in\Omega^2(M)$ which satisfies Maxwell's equations, which in absence of sources can be expressed as $dF=0$ and $\delta F=0$ where $d$ is the exterior derivative while $\delta$ is the codifferential. It is important to remark that, while the second equation depends on the metric associated to $M$ and hence on the geometry, the first one relies only on the smooth differentiable structure of the background and it is thus a constraint. It is at this stage that the scheme employed on Minkowski background encounters the first difficulties since, if we leave $M$ arbitrary and thus not a priori diffeomorphic to $\mathbb{R}^4$, we cannot apply Poincar\'e lemma to conclude the existence of $A\in\Omega^1(M)$ such that $F=dA$. In other words it is not true that it is always possible to reconstruct all field strengths fulfilling Maxwell's equations, even starting from an auxiliary object such as the vector potential $A$. The consequences of this result of differential geometry has far reaching consequences, since it tells us that, if one wants to quantize Maxwell's equations on a curved background, unless $M$ is somehow special, one cannot use $A$ as the building block, but has to refer uniquely to $F$. An example of a field strength which cannot be derived from the vector potential can be found in \cite{Ash80}. 

The aim of this paper is indeed to develop a quantization scheme for the field strength on an arbitrary four dimensional globally hyperbolic spacetime within the framework of the algebraic formulation of quantum field theory -- see \cite{Bongaarts} for an earlier investigation in this language. This is certainly not the first investigation in this direction and preliminary works are present in \cite{Kusku} and, as a special case of a much broader analysis, in \cite{Hollands}. Compared to these earlier publications, we rectify some minor problems mostly in the analysis of the classical dynamical system, but our main contribution will be related to the construction of a field algebra of observables for the field strength. In this endeavour we will prove in particular that the commutator between two generators of the algebra is given by the Lichnerowicz propagator \cite{Lich} regardless of the chosen spacetime. This allows us also to make a direct connection with an old result of Ashtekar and Sen \cite{Ash80}, who identified the existence of a two-parameter family of unitary inequivalent representations of the canonical commutation relations for the field strength on Schwarzschild spacetime. In our language this translates in the existence of a non trivial centre for the field algebra whenever the second de Rham cohomology group of the manifold, either with real or with complex coefficients, is non trivial. 

As a last point we can address the question whether the field strength tensor can be described as a locally covariant quantum field theory. First introduced in \cite{BFV}, the so-called principle of general local covariance was formulated leading to the realization of a quantum field theory as a covariant functor between the category of globally hyperbolic (four-dimensional) Lorentzian manifolds with isometric embeddings as morphisms and the category of $^*$-algebras with injective homomorphisms as morphisms. Already shown to hold true for scalars \cite{BFV}, spinors \cite{Sanders} and recently for the Proca field and for the vector potential (though in this case it has been assumed that the first de Rham cohomology group of the underlying background is trivial) \cite{Dappiaggi}, such paradigm turns out not to hold true in the case of a quantum field theory based on the field strength. Although we will be more explicit in the main body of the paper, we stress that the obstruction is related to a potential clash between the presence of a non trivial centre for the field algebra of $F$ in a globally hyperbolic spacetime $M$ and the existence of isometric embeddings of $M$ into backgrounds $M'$ with trivial second de Rham cohomology group. As a potential way out, we show that general local covariance can be restored if the category of admissible spacetimes is suitably reduced, although, as we shall comment later in detail, this has certainly far reaching physical consequences.

\vskip .2cm

The paper will be organized as follows: In section 1.1 we will recollect the notations and conventions we shall use throughout the paper. In section 2 we will instead discuss Maxwell's equations and the associated initial value problem, showing that it is well-defined and that the space of solutions can be constructed also in this case with the help of the causal propagator for a suitable second order hyperbolic differential operator. Section 3 will be instead entirely devoted to the issue of constructing the associated field algebra and, in particular, we shall prove that the commutator between two generators can be computed via the Lichnerowicz propagator. In section 3.1 we shall show that, whenever certain topological invariants of the background are not trivial, the field algebra possesses a non trivial centre and we fully characterize its elements also providing explicit examples. In section 3.2 we tackle the problem whether the principle of general local covariance holds true for the field strength finding in general a negative answer unless the class of admissible spacetimes is reduced. In section 4 we draw some conclusions. 

\subsection{Basic definitions and Conventions}\label{notations}

In this paper, each background will always be referred to as a {\em ``spacetime''}, that is a four dimensional differentiable, second countable, Hausdorff manifold $M$ with a Lorentzian metric $g$ whose signature is $(+,-,-,-)$. We shall also assume that $M$ is {\em globally hyperbolic}, hence there exists a closed achronal subset $\Sigma\subset M$ whose domain of dependence coincides with $M$ itself. Note, that a spacetime is usually assumed to be connected in the literature. Yet, in this paper and with respect to the construction of the field algebra in particular, we are also interested in disconnected manifolds with finitely many connected components. See in particular the remark after proposition \ref{universal algebra properties}. Note, that there is no ambiguity in carrying over established results from the connected case to the disconnected case with finitely many connected components. On account of standard results in differential geometry and of the recent analysis in \cite{Bernal, Bernal2} and of theorem 1.1 in \cite{Bernal3}, global hyperbolicity entails that there exists an isometry $\psi$ between $M$ and a smooth product manifold $\mathbb{R}\times\Sigma$. Thus $\Sigma$ turns out to be a three-dimensional embedded submanifold and theorem 2.1 in \cite{Bernal} guarantees, moreover, that $(\psi^{-1})^*g$ splits as $\beta d\mathcal{T}^2-h$ where $\mathcal{T}:\mathbb{R}\times\Sigma\to\mathbb{R}$ is a temporal function, $\beta\in C^\infty(\mathbb{R}\times\Sigma,(0,\infty))$ while $h$ induces for fixed values of $\mathcal{T}$ a smooth Riemannian metric on $\Sigma$. Furthermore global hyperbolicity yields that $M$ admits an orientation and thus, henceforth, we assume that a choice has been done and all spacetimes are globally hyperbolic as well as time oriented and oriented.

On top of the geometric structure we shall consider $\Omega^p(M,\mathbb{K})$ and $\Omega^p_0(M,\mathbb{K})$, respectively the set of smooth and of smooth and compactly supported $p$-forms on $M$ with values in the field $\mathbb{K}$, usually chosen either as $\mathbb{R}$ or $\mathbb{C}$. Here $p\geq 0$ and $\Omega^0_{(0)}(M,\mathbb{K}):= C^\infty_{(0)}(M,\mathbb{K})$, where the parenthesis around the subscript indicates that the statement holds true both with and without the subscript itself. Let $\mathbb{K}$ be the complex numbers; then, on these spaces, one can define two natural operators, the external derivative $d:\Omega^p_{(0)}(M,\mathbb{C})\to\Omega^{(p+1)}_{(0)}(M,\mathbb{C})$ and the Hodge dual $*:\Omega^p_{(0)}(M,\mathbb{C})\to\Omega^{(4-p)}_{(0)}(M,\mathbb{C})$. Notice that, while $d$ is completely independent from $g$, $*$ is instead a function of the underlying metric. Furthermore, since $*$ is invertible, we can introduce a third operator, known as the codifferential $\delta:=(-1)^p*^{-1}d*:\Omega^{p}_{(0)}(M,\mathbb{C})\to\Omega^{(p-1)}_{(0)}(M,\mathbb{C})$. 

In the main body of the paper we will be often interested in compactly supported smooth forms which are either closed or coclosed and to avoid to be redundant in the exposition we introduce the following novel notation:
\begin{eqnarray*}
\Omega^p_{0,\delta}(M,\mathbb{C}):=\left\{\omega\in\Omega^p_0(M,\mathbb{C})\;|\;\;\delta\omega = 0\right\},\\
\Omega^p_{0,d}(M,\mathbb{C}):=\left\{\omega\in\Omega^p_0(M,\mathbb{C})\;|\;\;d\omega = 0\right\}.
\end{eqnarray*}
To conclude, we mention two further ingredients we shall need in the forthcoming discussion. The first is $H^p(M,\mathbb{C})$ which is the $p$-th de Rham cohomology group of $M$ -- see \cite{BoTu} for the definition and for a recollection of the main properties. It is noteworthy that, since such groups are built only out of the external derivative, they are completely independent from the underlying geometry and from $g$ in particular. We can combine together $d$ and $\delta$ to define the Laplace-de Rham operator $\square:= -\left(d\delta+\delta d\right)$. The second ingredient is instead $H_p^\infty(M)$ which stands for the $p$-th smooth singular homology group of the manifold and whose main properties are discussed in \cite{Lee}.

\section{Maxwell's equations on curved spacetimes}
As stated in the introduction, the main objective of this paper is to shed some light on the classical and on the quantum structure of Maxwell's equations on curved backgrounds, emphasizing in particular how the underlying topology affects the qualitative behaviour of the system. To start with, we need to introduce the key objects of our analysis: The curved spacetime analogue of Maxwell's equations sees $F\in\Omega^2(M)$ as the dynamical variable and the dynamics is ruled by
\begin{equation}\label{Maxwell}
dF=0,\qquad-\delta F= j,
\end{equation}
where $j\in\Omega^1(M)$ is the external current such that $\delta j=0$. A key property of \eqref{Maxwell} when defined on a generic spacetime $(M,g)$ lies in the analysis of the first identity. This is a constraint on the form of $F$ which usually leads to state both that there exists $A\in\Omega^1(M)$ such that $F=dA$ and that one can consider $A$, the so-called vector potential as the underlying dynamical field. This statement is based on the Poincar\'e lemma which, alas, cannot be always applied since it fails to hold true whenever $H^2(M)$ is not trivial. In this particular case, it turns out that there exist classical field strengths which cannot be derived as the external derivative of a suitable one-form. Since, from a physical point of view, it is $F$ the observable field of the dynamical system, it is natural to wonder whether a full classical and quantum analysis of \eqref{Maxwell} could be performed without making use of any auxiliary structure such as the vector potential. 

In order to grasp the classical behaviour of a dynamical system ruled by \eqref{Maxwell}, we need to prove that this set of equations admits a well-defined initial value problem on every globally hyperbolic spacetime. Despite the apparent obviousness of this question, to the best of our knowledge it turns out that this problem has been only partly treated in details and the discussions available in the literature are either partly incomplete or based upon further restrictive assumptions, such as the compactness of the Cauchy surface $\Sigma$ -- see \cite{Kusku, Hollands}, but also \cite{Dimock, Fewster} although they work with the vector potential. On the opposite, since we want to cope with the most general scenario, we need the following statement -- see also \cite{Pfenning} for a similar analysis:
\begin{prop}\label{Cauchy problem field strength tensor}
Let \((M, g)\) be an oriented 4-dimensional globally hyperbolic spacetime with at most finitely many connected components whose smooth spacelike Cauchy surface is $\Sigma$ with smooth embedding $\iota:\Sigma\hookrightarrow M$. Then, for each triple $(j, E, B)$ such that \(j\in\Omega^1_{0,\delta}(M, \mathbb{C})\), \(E\in\Omega^1_0(\Sigma, \mathbb{C})\) with \(-\delta E=*_{(3)}\iota^{pb}*j\) and \(B\in\Omega^2_{0,d}(\Sigma, \mathbb{C})\), there exists a unique solution \(F\in\Omega^2(M, \mathbb{C})\) of the initial value problem\footnote{Notice that, in order to avoid a potential confusion in the notation, in this section, we refer to the pull-back of $\iota$ as $\iota^{pb}$ in place of $\iota^*$. Furthermore we indicate with $*_{(3)}$ the Hodge dual induced on the Cauchy surface $\Sigma$ to distinguish it from the one on $M$.}
\begin{equation}\label{MaxCauchy}\left\{\begin{array}{ll}
dF=0,&
-\delta F= j,\\
-*_{(3)}\iota^{pb}*F=E, &
-\iota^{pb}F=B.\end{array}
\right. 
\end{equation}
Furthermore, \(F\) depends linearly and continuously on both the source term \(j\), and on the initial data \(E\), \(B\). Each solution also enjoys the following support property:
\begin{align*}
\supp(F)\subseteq J^+\big(X\big)\cup J^-\big(X\big),
\end{align*}
where $J^\pm(X)$ are the causal future and past of $X:= \supp(j)\cup\supp(E)\cup\supp(B)$ respectively.
\end{prop}

\begin{proof}
We start with the connected case. Since $\Box=-(d\delta+\delta d)$, one can realize by direct inspection that every solution of \eqref{MaxCauchy} also solves  $\Box F=dj$. Yet, in order to use the latter as a starting point to solve Maxwell's equations, we need to prove that we can always select suitable initial data for the wave equation so that a solution of the latter yields also one of \eqref{MaxCauchy}. To this avail, let us consider $\mathcal{F},\Pi\in\iota^{pb}\Omega^2_0(M, \mathbb{C})$ where $\iota^{pb}$ here is the pull-back induced from $\iota:\Sigma\hookrightarrow M$ on the compactly supported sections of the exterior bundle on $M$. In other words $\mathcal{F}$ and $\Pi$ are maps from $\Sigma$ into $\Omega^2_0(M,\mathbb{C})$ such that  
\begin{align*}
\left.\mathcal{F}\right|_{V\cap\Sigma}&=n_0E_jd\phi^0\wedge d\phi^j-\frac{1}{2}B_{ij}d\phi^i\wedge d\phi^j,\\
\left.\Pi\right|_{V\cap\Sigma}&=n^0\nabla\!_i\mathcal{F}_{0j}d\phi^i\wedge d\phi^j+n_0(j_k-g^{ij}\nabla\!_i\mathcal{F}_{jk})d\phi^0\wedge d\phi^k.
\end{align*}
On account of $M$ being isometric to $\mathbb{R}\times\Sigma$ with line element $ds^2=\beta d\mathcal{T}^2-h$ as outlined in section \ref{notations}, here $V$ is a coordinate patch of $M$ adapted to this last metric. It intersects $\Sigma$ on a non empty open set and it is endowed with a local chart $\phi$, whereas $n_\mu$ is the unit normal vector to $\Sigma$. Hence, the Cauchy problem
\begin{equation*}\left\{\begin{array}{ll}
\Box F =dj, & \\
\left.F\right|_\Sigma=\mathcal{F}, & \left.\nabla_\mathfrak{n}F\right|_\Sigma=\Pi\end{array}\right. ,
\end{equation*}
where both $\mathcal{F}$ and $\Pi$ are chosen as in the previously displayed set of two equations, admits a unique solution $F\in\Omega^2(M,\mathbb{C})$ which, furthermore, on account of \cite[Thm.3.2.11]{BGP}, depends linearly and continuously  both on the source term and on the initial data \(\mathcal{F}\), \(\Pi\). At the same time it holds \(\supp(F)\subset J^+\big(X\big)\cup J^-\big(X\big)\) where $X:=\supp(dj)\cup\supp(\mathcal{F})\cup\supp(\Pi)$, which, in turn, entails the sought support property.  It remains to be shown that the obtained solution \(F\) of the Cauchy problem for the wave equation solves \eqref{MaxCauchy} as well. To achieve this, it suffices to show that $F$ also satisfies 
\begin{align*}
\Box dF=0, & & \Box(-\delta F+j)=0,
\end{align*}
with vanishing initial data. Since $[\Box, d]=[\Box,\delta]=0$, the two equations automatically descend from $\Box F=dj$ and thus only the initial data have to be checked.  It suffices to show it in \((V, \phi)\). From \(dB=0\) and \(-\iota^{pb}F=B\) it follows directly \((\nabla\!_kF_{ij}+\nabla\!_jF_{ki}+\left.\nabla\!_iF_{jk})\right|_{V\cap\Sigma}=0\), whereas, from \(\left.\nabla\!_\mathfrak{n}F\right|_{V\cap\Sigma}=\left.\Pi\right|_{V\cap\Sigma}\), it descends \((\nabla\!_0F_{ij}+\nabla\!_jF_{0i}+\nabla\!_i\left.F_{j0})\right|_{V\cap\Sigma}=0\); hence \(\left.dF\right|_{V\cap\Sigma}=0\). Equivalently \(\nabla\!_\mathfrak{n}\left.F\right|_{V\cap\Sigma}=\left.\Pi\right|_{V\cap\Sigma}\) yields \((n^\nu\nabla^\mu F_{\mu\nu}-\left.n^\nu j_\nu)\right|_{V\cap\Sigma}=0\). Notice that \((\nabla^\mu \left.F_{\mu 0}-j_0)\right|_{V\cap\Sigma}=0\) is a by-product of both \(-\delta E=*_{(3)}\iota^{pb}*j\) and \(-*_{(3)}\iota^{pb}*F=E\); thus \(\left.(-\delta F-j)\right|_{V\cap\Sigma}=0\). The remaining initial condition \(\nabla\!_\mathfrak{n}\left.dF\right|_{V\cap\Sigma}=0\) arises out of \(\Box F=dj\), of the properties of \([\nabla_\mu, \nabla\!_\nu]\) and of the symmetries of the Riemannian curvature tensor -- see \cite{Lang}. Hence, on account \cite[Cor.3.2.4]{BGP}, this suffices for \(dF=0\) to hold true on  \(M\). To conclude, \(\big(\nabla\!_\mathfrak{n}\left.(-\delta F-j)\big)\right|_{V\cap\Sigma}=0\) is a result of \(dF=0\), \(\Box F=dj\) and of the conservation of the current \(\delta j=0\). As above this suffices to prove \(-\delta F=j\) on $M$. To conclude, in order to establish that the solution of \eqref{MaxCauchy}, it suffices to suppose that there exists two different solutions, say $F$ and $F'$. Their difference $\widetilde{F}\doteq F-F'$ must satisfy $d\widetilde{F}=0$ and $\delta\widetilde{F}=0$ with vanishing initial data on the Cauchy surface $\Sigma$. This entails that $\widetilde{F}$ must also satisfy $\Box\widetilde{F}=0$ with vanishing inital data on $\Sigma$ and, according to standard results of the theory of partial differential equations, this holds true only if $\widetilde{F}=0$, hence $F=F'$. If $M$ is disconnected with finitely many connected components $\Gamma_1,\dots\Gamma_n$, $n\in\mathbb{N}$, we consider the partition of unity $\{\chi^i\}_{i=1,\dots, n}$ subordinated to $\{\Gamma_i\}_{i=1,\dots, n}$ such that $\chi^i\bigl|_{\Gamma_i}=1$ and $\chi^i\bigl|_{\Gamma_j}=0$ for $i\not=j$. Since $(M,g)$ is globally hyperbolic, so is $(\Gamma_i, g\bigl|_{\Gamma_i})$, with smooth spacelike Cauchy surface $\Sigma_i=\Sigma\cap\Gamma_i$; we obtain, therefore, the initial value problem
\begin{equation*}\left\{\begin{array}{ll}
dF_i=0,&
-\delta F_i= \chi^ij,\\
-*_{(3)}\iota^{pb}*F_i=\chi^iE, &
-\iota^{pb}F_i=\chi^iB.\end{array}
\right.
\end{equation*}
We solve each of these initial value problems as prescribed in the connected case, therefore obtaining unique solutions $F_i$. The unique solution $F$ is then assembled via the partition of unity, $F=\sum_{i=1}^n\chi^iF_i$.
\end{proof}

As a by-product of this last proposition, we can construct the solutions of Maxwell's equations on a globally hyperbolic spacetime starting from the wave equation. If we focus on the source free case, that is $j=0$, we can generate all solutions of $\Box F=0$ with compactly supported initial data as $F=G\omega$ where $\omega\in\Omega^2_0(M, \mathbb{C})$ and where $G:= G^+-G^-$ is the causal propagator \cite{BGP}. Here $G^\pm:\Omega^2_0(M, \mathbb{C})\to\Omega^2(M, \mathbb{C})$ are the uniquely defined advanced and retarded Green operators such that $G^\pm\circ\Box=\Box\circ G^\pm=id_{\Omega^2_0(M, \mathbb{C})}$ and $\supp\left(G^\pm(\omega)\right)\subseteq J^\pm\left(\supp(\omega)\right)$, for all $\omega\in\Omega^2_0(M, \mathbb{C})$. Notice that these properties of $G^\pm$ also entail that every compactly supported smooth solution of $\Box F=0$ must vanish identically. An additional noteworthy property of the causal propagator associated to the Laplace-de Rham operator originates from its structure and from the fact that $\square$ intertwines between the codifferential operator $\delta$ acting on $p$ and on $(p-1)$-forms, that is $\square\circ\delta=\delta\circ\square$. To wit, at a level of solutions of the corresponding wave equation with smooth and compactly supported initial data, this entails that, if $F=G(\omega)$ with $\omega\in\Omega^2_0(M, \mathbb{C})$, then $\delta F= \delta G(\omega)=G(\delta\omega)$. The very same properties hold with respect to the exterior derivative $d$. 

Since not all $G(\omega)$ fulfil also the source free Maxwell's equations, one needs to impose some further constraints on the set of initial test functions $\omega$ in order to take into account only the two-forms solving \eqref{MaxCauchy}. The following proposition amends this deficiency:      

\begin{prop}
A smooth complex $2$-form $F$ is a solution of \eqref{MaxCauchy} with $j=0$ and with compactly supported smooth initial data if and only if there exist \(\alpha\in\Omega^3_{0,d}(M, \mathbb{C})\) and \(\beta\in\Omega^1_{0,\delta}(M, \mathbb{C})\) such that $F=G(\delta\alpha+d\beta)$.
\end{prop}

\begin{proof}
``\(\Longleftarrow\)'' Since \(\alpha\) and \(\beta\) are of compact support and since G commutes with \(d\) and \(\delta\), it holds that $dF=G(d\delta\alpha)=-G(\Box\alpha)=0$ and that $\delta F=G(\delta d\beta)=-G(\Box\beta)=0$. Furthermore, on account of the properties of the causal propagator, it is also guaranteed that the initial data of Maxwell's equations associated to $G(d\alpha+\delta\beta)$ are smooth, compactly supported and their form fulfils the constraints of \eqref{MaxCauchy}.\\
``\(\Longrightarrow\)'' Since \(dF=0\) and \(\delta F=0\) entail \(\Box F=0\), there must exist \(\omega\in\Omega^2_0(M, \mathbb{C})\) such that \(F=G\omega\). Furthermore, \(dF=dG\omega=Gd\omega=0\) and \(\delta F=\delta G\omega=G\delta\omega=0\) entail the existence of \(\alpha\in\Omega^3_{0,d}(M, \mathbb{C})\) and \(\beta\in\Omega^1_{0,\delta}(M, \mathbb{C})\) satisfying \(d\omega=\Box\alpha\) and \(\delta\omega=\Box\beta\). On account of the nilpotency of both $d$ and $\delta$, it holds $\Box d\alpha=0$ and $\Box\delta\beta=0$ which suffices to conclude that $d\alpha=\delta\beta=0$, $\alpha$ and $\beta$ being compactly supported.
The same reasoning entails that the following chain of identities \(\Box\omega=-d\delta\omega-\delta d\omega=-\Box d\beta-\Box\delta\alpha\) yields $\omega=-\delta\alpha-d\beta$, up to an irrelevant sign the sought result.
\end{proof}

\section{Quantisation of the field strength tensor}
The full control of the classical dynamics of Maxwell's equations allows us to address the problem of quantising a field theory with $F$ as the main ingredient. As it is customary in the algebraic approach, this is a two-step procedure, the first calling for the identification of a suitable algebra of observables and the second requiring the assignment of a state to represent such an algebra in terms of operators on a suitable Hilbert space. In this paper we will focus on the first part of the programme, hence we shall construct the full field algebra and investigate its properties. In the process we will benefit from ideas which first appeared in earlier works \cite{Fredenhagen, Fredenhagen2} and \cite{Hollands}; the sketch of the construction is the following: First we consider a suitable covering of $(M,g)$ in globally hyperbolic submanifolds \((M_i, g_i)\), \(i\in I\) where $I$ is a set, which we will specify below. Afterwards, we construct the local field algebras \(\mathfrak{F}(M_i)\) of the field strength tensor, whereas the global one is defined as the universal algebra \(\mathfrak{F}_u(M)\) associated to the local algebras \(\mathfrak{F}(M_i)\). The commutation relations encoded in \(\mathfrak{F}_u(M)\) will turn out to be given by the Lichnerowicz's commutator and the algebra itself will have all the properties required to deserve the name ``\emph{global}" field algebra. One could wonder why it is necessary to go through such an involute construction. There are many conceptual reasons but it is noteworthy that the form of the commutator is in principle only known for spacetimes with certain topological restrictions and thus we need to show that a generalization to a more generic spacetime exists. 

\subsection{The universal algebra}
\begin{thm}\label{universal algebra}
Let $I$ be a set and $(A_i)_{i\in I}$ a family of unital $*$-algebras together with linking unital $*$-homomorphisms $\alpha_{ij}:A_i\longrightarrow A_j$ where the admissible pairs $(i,j)$ are a suitable subset of $I\times I$. Furthermore the compatibility condition $\alpha_{jk}\circ\alpha_{ij}=\alpha_{ik}$ holds true whenever $\alpha_{jk}$ and $\alpha_{ij}$ are defined. Then, there exists a unique (up to $*$-isomorphism) unital $*$-algebra $A_u$ together with a family of unital $*$-homomorphisms $(\alpha_i:A_i\longrightarrow A_u)_{i\in I}$ such that the following universal property holds true:
\begin{description}
\item[(UVA)] For each unital $*$-algebra $B$ and for each family $(\phi_i:A_i\longrightarrow B)_{i\in I}$ of unital $*$-homomorphisms such that $\phi_j\circ\alpha_{ij}=\phi_i$ whenever $\alpha_{ij}$ exists, there exists a unique unital $*$-homomorphism $\Phi_u:A_u\longrightarrow B$ which satisfies $\Phi_u\circ\alpha_i=\phi_i$ for all $i\in I$.
\end{description}
The pair $\big(A_u, (\alpha_i)_{i\in I}\big)$ is called  the {\bf universal algebra} of $\big((A_i)_{i\in I}, (\alpha_{ij})_{(i,j)}\big)$. 
\end{thm}

\begin{proof}
To begin with, we regard the unital $*$-algebras $A_i$, $i\in I$, as complex vector spaces and construct the associative tensor algebra $\mathcal{T}(\bigoplus_{i\in I}A_i)$ of their direct sum. With componentwise addition, componentwise multiplication with a scalar, componentwise antilinear involution $*$ and multiplication induced by the algebraic tensor product $\otimes$, $\mathcal{T}(\bigoplus_{i\in I}A_i)$ becomes a unital $*$-algebra. Next, we consider the two-sided $*$-ideal $\mathcal{I}$ generated by those elements of the form
\begin{align*}
&\Big(0_\mathbb{C},-(a^i_1a^i_2)^j_{i\in I}, (a^i_1)^j_{i\in I}\otimes(a^i_2)^j_{i\in I}, 0_{(\bigoplus_{i\in I}A_i)^{\otimes 3}},\dots\Big),\\
&\Big(1_\mathbb{C},-(1_{A_i})^j_{i\in I}, 0_{(\bigoplus_{i\in I}A_i)^{\otimes n}},\dots\Big),\\
&\Big(0_{\mathbb{C}},\big(\alpha_{ik}(a^i)\big)^j_{k\in I})-(a^k)^i_{k\in I},0_{(\bigoplus_{i\in I}A_i)^{\otimes 2}},\dots\Big),
\end{align*}
$a^i_1,a^i_2,a^i\in A_i$ and for all given $\alpha_{ij}$. $(a^i)^j_{i\in I}$ denotes the vector in $\bigoplus_{i\in I}A_i$ for which every entry is zero except the $j$-th one which is precisely $a^j\in A_j$. We denote the equivalence class of an element $a\in\mathcal{T}(\bigoplus_{i\in I}A_i)$ with respect to that quotient by $[a]$. Now, define
\begin{align*}
A_u:=\mathcal{T}(\bigoplus_{i\in I}A_i)\big/\mathcal{I}
\end{align*}
and for $j\in I$
\begin{align*}
\alpha_j:A_j\longrightarrow A_u\qquad a^j\longmapsto\Big[0_\mathbb{C}, (a^i)^j_{i\in I},0_{(\bigoplus_{i\in I}A_i)^{\otimes 2}},\dots\Big]. 
\end{align*}
$A_u$ defined in this way is a unital $*$-algebra and, per direct inspection, $\alpha_i$ turns out to be a well-defined unital $*$-homomorphism for all $i\in I$. The pair $\big(A_u, (\alpha_i)_{i\in I}\big)$ satisfies the universal property. To wit, let $B$ be any arbitrary unital $*$-algebra and $(\phi_i:A_i\longrightarrow B)_{i\in I}$ a family of unital $*$-homomorphisms such that $\phi_j\circ\alpha_{ij}=\phi_i$, whenever $\alpha_{ij}$ exists. Since any element $[a]\in A_u$ can be written as
\begin{align*}
[a]=\underset{n\in\mathbb{N}}{\sum}T_n^{k_1\dots k_n}\underset{l=1}{\overset{n}{\prod}}\underset{j\in I}{\sum}\Big[0_\mathbb{C},(a_{k_l}^i)^j_{i\in I}, 0_{(\bigoplus_{i\in I}A_i)^{\otimes 2}},\dots\Big],
\end{align*}
$T_n^{k_1\dots k_n}\in\mathbb{C}$, by the structure of $\mathcal{T}(\bigoplus_{i\in I}A_i)$, a unital $*$-homomorphism $\Phi_u:A_u\longrightarrow B$ such that $\Phi_u\circ\alpha_i=\phi_i$ for all $i\in I$ is uniquely fixed by 
\begin{align*}
\Phi_u([a])=\underset{n\in\mathbb{N}}{\sum}T_n^{k_1\dots k_n}\underset{l=1}{\overset{n}{\prod}}\underset{j\in I}{\sum}\phi_j(a^j_{k_l}).
\end{align*}
This shows the existence of the universal algebra. Let $\big(B, (\beta_i)_{i\in I}\big)$ be another pair consisting of a unital $*$-algebra and unital $*$-homomorphisms having the universal property (UVA). Thus we have a unique unital $*$-homomorphism $\Psi:B\longrightarrow A_u$ fulfilling $\Psi\circ\alpha_i=\beta_i$ for all $i\in I$. According to (UVA), $\Phi_u$ is the unique unital $*$-homomorphism such that $\Phi_u\circ\beta_i=\alpha_i$ for all $i\in I$. $A_u\overset{\Psi}{\longrightarrow}B\overset{\Phi_u}{\longrightarrow}A_u$ and $\Phi_u\circ\Psi\circ\alpha_i=\alpha_i$ for all $i\in I$. However, since $\big(A_u, (\alpha_i)_{i\in I}\big)$ has the universal property, thus the unital $*$-homomorphism $A_u\longrightarrow A_u$ is unique, and $\id_{A_u}:A_u\longrightarrow A_u$ satisfies $\id_{A_u}\circ\alpha_i=\alpha_i$ for all $i\in I$ as well, necessarily it holds that $\Phi_u\circ\Psi=\id_{A_u}$. $B\overset{\Phi_u}{\longrightarrow}A_u\overset{\Psi}{\longrightarrow}B$ and $\Psi\circ\Phi_u\circ\beta_i=\beta_i$ for all $i\in I$. In the same manner, since $\big(B, (\beta_i)_{i\in I}\big)$ has the universal property, the unital $*$-homomorphism $B\longrightarrow B$ is unique and $\id_B:B\longrightarrow B$ satisfies $\id_B\circ\beta_i=\beta_i$ for all $i\in I$, $\Psi\circ\Phi_u=\id_B$. Thereby $A_u$ and $B$ are isomorphic via a unital $*$-isomorphism.
\end{proof}
\subsection{Tiling the spacetime}
Let us recall that every connected component $\Gamma_c$, $c=1,\dots,n$ for an $n\in\mathbb{N}$, of an oriented 4-dimensional globally hyperbolic spacetime \((M, g)\) with at most finitely many connected components will turn into a connected, oriented, 4-dimensional, globally hyperbolic, embedded subspacetime, if endowed with the structures induced by $(M, g)$. Consequently any of these $\Gamma_c$ can be foliated up to isometries as \(\mathbb{R}\times\Sigma_c\). Here $\Sigma_c=\Sigma\cap\Gamma_c$ being a smooth spacelike Cauchy surface endowed with the natural structures inherited from $(\Gamma_c,g_c)$, $\iota_\Sigma:\Sigma\longrightarrow M$ is a smooth spacelike Cauchy surface of $(M, g)$. Therefore each $x\in M$ lies precisely in one connected component $\Gamma_c$ and at least on one of such surfaces, which we denote by $\Sigma_{cx}$, and we can always construct an open subset $S_x\subseteq\Sigma$ that is either contractible or disconnected with finitely many contractible connected components. The net advantage is that its associated Cauchy development $D^M(S_x)$ is in turn a contractible open subset or a disconnected open subset of $M$ with finitely many contractible connected components that will become an oriented, 4-dimensional, globally hyperbolic, embedded subspacetime, if equipped with the structures induced by $(M, g)$. Since this procedure can be repeated for all points of the manifold, we can always cover $M$ with contractible open subsets and disconnected open subsets with finitely many contractible connected components such that these open subsets will become oriented, 4-dimensional, globally hyperbolic, embedded subspacetimes, if endowed with the structures induced by $(M, g)$. But not any such cover will do the trick. We need a very specific cover, namely $\bigcup_{i\in I}M_i=M$ that of all contractible open subsets of $M$ and all disconnected open subsets of $M$ with finitely many contractible connected components such that $M_i$ becomes an oriented, 4-dimensional, globally hyperbolic, embedded subspacetime for all $i\in I$, if endowed with the structures induced by $(M, g)$ and such that in addition the image of the inclusion $\iota_i:M_i\longrightarrow M$ is causally convex. The endpoint are oriented 4-dimensional globally hyperbolic spacetimes with at most finitely many contractible connected components, which we denote as $(M_i, g_i)$, $i\in I$. $I$ is actually a set because this cover $\{M_i\mid i\in I\}$ is contained in the power set of $M$. 
\subsection{The local field algebras}
For each $i\in I$, we associate to the oriented 4-dimensional globally hyperbolic spacetime $(M_i, g_i)$ with at most finitely many contractible connected components the local field algebra $\mathfrak{F}(M_i)$ of $F$. Notice that, since $M_i$ has at most finitely many contractible connected components contractible, the first equation in \eqref{Maxwell} entails via Poincar\'e lemma that $F=dA$ where $A\in\Omega^1(M_i,\mathbb{C})$. Also, there is no ambiguity in constructing the field algebra for a disconnected spacetime with finitely many contractible connected components in the same manner as it is done in the connected contractible case.
\begin{defi}\label{fieldal}
We call the field algebra of the field strength tensor on an oriented 4-dimensional globally hyperbolic spacetime $(M,g)$ with at most finitely many contractible connected components, $\mathfrak{F}(M)$, the unital $*$-algebra generated by the elements $\widehat{\bf F}(\omega)$ with $\omega\in\Omega^2_0(M,\mathbb{C})$ together with the defining relations 
\begin{equation*}
\label{properties local field strength operator}\begin{array}{l}
\textrm{EOM 1)}\quad \widehat{\bf F}(\omega)=0,\quad\forall\omega=\delta\eta,\quad\eta\in\Omega^3_0(M,\mathbb{C})\\
\textrm{EOM 2)}\quad \widehat{\bf F}(\omega)=0,\quad\forall\omega=d\theta,\quad\theta\in\Omega^1_0(M,\mathbb{C})\\
\textrm{COMM)}\quad \Big[\widehat{\bf{F}}(\omega), \widehat{\bf{F}}(\omega')\Big]=i\,(\int_M G\delta\omega\wedge*\delta\omega')\,1_{\mathfrak{F}(M)},\enspace\forall\omega, \omega'\in\Omega^2_0(M, \mathbb{C}),\\
\end{array}
\end{equation*}
where $G$ is the causal propagator associated to the $\Box$-operator and $1_{\mathfrak{F}(M)}$ is the identity element of the algebra. The $*$-operation is the complex conjugation.
\end{defi}
We remark that, in the above definition, the first two conditions entail the fulfillment of the equations of motion and the equalities are meant in a distributional sense, {\it i.e.}, $\widehat{\bf F}(\delta\eta)=d\widehat{\bf F}(\eta)=0$ and similarly for {\em EOM 2)}. The form of the commutator descends from earlier analyses, see in particular \cite{Lich, Dimock}. Notice also that isotony is automatically implemented, that is, for given $(M_1,g_1)$ and $(M_2,g_2)$ with $M_1\subseteq M_2$ and $g_1=g_2|_{M_1}$, then $\mathfrak{F}(M_1)\subseteq\mathfrak{F}(M_2)$. In other words there always exists an injective unital $*$-homomorphism of algebras \(\alpha_{12}:\mathfrak{F}(M_1)\longrightarrow\mathfrak{F}(M_2)\), realised by $\alpha_{12}\big(\widehat{\textbf{F}}_1(\omega)\big):=\widehat{\textbf{F}}_2(\iota_{12*}\omega)$ with the help of the inclusion $\iota_{12}:M_1\longrightarrow M_2$ and subject to the additional compatibility condition \(\alpha_{23}\circ\alpha_{12}=\alpha_{13}\), whenever we consider three oriented 4-dimensional globally hyperbolic spacetimes with at most finitely many contractible connected components such that \(M_1\subseteq M_2\subseteq M_3\). Consequently, we obtain a family $\big(\mathfrak{F}(M_i)\big)_{i\in I}$ of unital $*$-algebras together with linking unital $*$-homomorphisms $\alpha_{ij}:\mathfrak{F}(M_i)\longrightarrow\mathfrak{F}(M_j)$ for $M_i\subseteq M_j$ that meet the compatibility condition $\alpha_{jk}\circ\alpha_{ij}=\alpha_{ik}$ whenever $M_i\subseteq M_j\subseteq M_k$.
\subsection{The global field algebra and as the universal algebra}
Since the system of unital $*$-algebras together with the unital $*$-homomorphisms as specified before satisfies the conditions of theorem \ref{universal algebra}, its application yields the universal algebra $\big(\mathfrak{F}_u(M), (\alpha_i)_{i\in I}\big)$. We define the global field algebra of the field strength tensor of an arbitrary oriented 4-dimensional globally hyperbolic spacetime $(M,g)$ with at most finitely many connected components to be the unital $*$-algebra $\mathfrak{F}_u(M)$. Its properties are clarified in the following statements:

\begin{lem} $\mathfrak{F}_u(M)$ satisfies the local compatibility condition
\begin{align*}
\alpha_i\big(\widehat{\textbf{F}}_i(\iota^*_i\omega)\big)=\alpha_{j}\big(\widehat{\textbf{F}}_j(\iota^*_j\omega)\big)\in\mathfrak{F}_u(M)
\end{align*}
whenever $\omega\in\Omega^2_0(M,\mathbb{C})$ such that $\supp\omega\subset M_i\cap M_j$.
\end{lem}

\begin{proof}
Since $(M,g)$ is globally hyperbolic, its standard topology coincides with the Alexandrov one, i.e. that which has the diamonds $I^M_-(p)\cap I^M_+(q)$ as its basis. As a result of that, there exists for every point $x\in\supp\omega$ a contractible diamond $D_x$ containing $x$ and $D_x\subseteq M_i\cap M_j$, in particular $D_x\subseteq M_i,M_j$. With the structures induced by $(M,g)$ these diamonds $D_x$ become 4-dimensional oriented globally hyperbolic embedded subspacetimes and therefore belong to our chosen cover $\bigcup_{i\in I}M_i=M$. Let $(\chi^x)_{x\in\supp\omega}$ be a partition of unity subordinated to that open cover $\bigcup_{x\in\supp\omega}D_x\supseteq\supp\omega$. Since a partition of unity is locally finite, all appearing sums are actually finite and we can compute
\begin{align*}
\alpha_i\big(\widehat{\textbf{F}}_i(\iota^*_i\omega)\big)&=\underset{x}{\sum}\alpha_i\big(\widehat{\textbf{F}}_i(\iota_{xi*}\iota^*_{xi}\iota^*_i\chi^x\omega)\big)=\underset{x}{\sum}\alpha_i\circ\alpha_{xi}\Big(\widehat{\textbf{F}}_x\big((\iota_i\circ\iota_{xi})^*\chi^x\omega\big)\Big)\\
&=\underset{x}{\sum}\alpha_x\big(\widehat{\textbf{F}}_x(\iota^*_x\chi^x\omega)\big)=\underset{x}{\sum}\alpha_j\circ\alpha_{xj}\Big(\widehat{\textbf{F}}_x\big((\iota_j\circ\iota_{xj})^*\chi^x\omega\big)\Big)\\
&=\alpha_j\big(\widehat{\textbf{F}}_j(\iota^*_j\omega)\big),
\end{align*}
which is the sought result.
\end{proof}

\begin{prop}\label{universal algebra properties} In $\mathfrak{F}(M)$, we can define \emph{global} smeared field strength operators $\widehat{\textbf{F}}(\omega)$ for all $\omega\in\Omega^2_0(M,\mathbb{C})$ such that
\begin{enumerate}
\item [$\bullet$] $\mathfrak{F}(M)$ is generated by the global smeared field strength operators $\widehat{\textbf{F}}(\omega)$, $\omega\in\Omega^2_0(M, \mathbb{C})$,
\item [$\bullet$] $\widehat{\textbf{F}}$ fulfils Maxwell's equations in a weak sense, i.e. EOM 1)
$\widehat{\textbf{F}}(\delta\eta)=0$ for all $\eta\in\Omega^3_0(M, \mathbb{C})$ and EOM 2) $\widehat{\textbf{F}}(d\theta)=0$ for all $\theta\in\Omega^1_0(M, \mathbb{C})$,
\item [$\bullet$] $\widehat{\textbf{F}}(z_1\omega_1+z_2\omega_2)=z_1\widehat{\textbf{F}}(\omega_1)+z_2\widehat{\textbf{F}}(\omega_2)$ for all $z_i\in\mathbb{C}$, for all $\omega_i\in\Omega^2_0(M, \mathbb{C})$ (linearity),
\item [$\bullet$] $\widehat{\textbf{F}}(\omega)^*=\widehat{\textbf{F}}(\overline{\omega})$ for all $\omega\in\Omega^2_0(M, \mathbb{C})$ (Hermicity).
\end{enumerate}
Furthermore, $\mathfrak{F}(M)$ obeys the principle of locality, i.e. $\big[\widehat{\textbf{F}}(\omega),\widehat{\textbf{F}}(\omega')\big]=0$ for all $\omega,\omega'\in\Omega^2_0(M,\mathbb{C})$ that are spacelike separated.
\end{prop}

\begin{proof}
Choose any partition of unity $(\psi^i)_{i\in I}$ subordinated to the cover $\bigcup_{i\in I}M_i=M$ and define
\begin{align*}
\widehat{\textbf{F}}(\omega):=\underset{i\in I}{\sum}\alpha_i\big(\widehat{\textbf{F}}_c(\iota^*_i\psi^i\omega)\big).
\end{align*}
First of all, the sum is finite because the partition of unity is locally finite and $\omega$ is of compact support. Secondly, this definition does not depend on the chosen partition of unity for let $(\varphi^j)_{j\in J}$ be another partition of unity, then
\begin{align*}
\widehat{\textbf{F}}(\omega)&=\underset{i\in I}{\sum}\alpha_i\big(\widehat{\textbf{F}}_i(\iota^*_i\psi^i\omega)\big)=\underset{i,j\in I}{\sum}\alpha_i\big(\widehat{\textbf{F}}_i(\iota^*_i\varphi^j\omega^i)\big)=\underset{i,j\in I}{\sum}\alpha_j\big(\widehat{\textbf{F}}_j(\iota^*_j\varphi^j\omega^i)\big)\\
&=\underset{j\in I}{\sum}\alpha_j\big(\widehat{\textbf{F}}_j(\iota^*_j\varphi^j\omega)\big)
\end{align*} 
where we applied the local compatibility of $\mathfrak{F}(M)$, i.e. the foregoing lemma. Note, that this implies another local compatibility property for $\mathfrak{F}(M)$, namely let $\omega\in\Omega^2_0(M,\mathbb{C})$ be completely contained in one $M_j$ for a $j\in I$, i.e. $\supp\omega\subseteq M_j$, then
\begin{align*}
\widehat{\textbf{F}}(\omega)&=\underset{i\in I}{\sum}\alpha_i\big(\widehat{\textbf{F}}_i(\iota^*_i\psi^i\omega)\big)=\underset{i\in I}{\sum}\alpha_j\big(\widehat{\textbf{F}}_j(\iota^*_j\psi^i\omega)\big)=\alpha_j\big(\widehat{\textbf{F}}_j(\iota^*_j\omega)\big).
\end{align*}
Linearity, Hermicity and Maxwell's equations in a weak sense follow from this definition and from their implementation at a level of local field algebras $\mathfrak{F}(M_i)$, $i\in I$. To show locality, we start with a slightly simpler statement. Let $\omega\in\Omega^2_0(M,\mathbb{C})$ and $\omega'\in\Omega^2_0(M, \mathbb{C})$ be spacelike separated and $\supp\omega\subseteq M_i$ for a $i\in I$ and $\supp\omega'\subseteq M_j$ for a $j\in I$ such that $M_i$ and $M_j$ are spacelike separated, i.e. $M_i\cap M_j=\emptyset$. $M_i\sqcup M_j$ can be regarded as an oriented, 4-dimensional, globally hyperbolic, embedded submanifold of $(M,g)$ with finitely many contractible connected components. Hence, we have an injective unital $*$-homomorphism $\alpha_{i\sqcup j}:\mathfrak{F}(M_i\sqcup M_j)\longrightarrow\mathfrak{F}_u(M)$ for $i\sqcup j\in I$. We find that
\begin{align*}
\Big[\widehat{\textbf{F}}_u(\omega), \widehat{\textbf{F}}_u(\omega')\Big]&=\Big[\alpha_{i\sqcup j}\big(\widehat{\textbf{F}}_{i\sqcup j}(\iota^*_{i\sqcup j}\omega)\big), \alpha_{i\sqcup j}\big(\widehat{\textbf{F}}_{i\sqcup j}(\iota^*_{i\sqcup j}\omega')\big)\Big]\\
&=\alpha_{i\sqcup j}\Big[\widehat{\textbf{F}}_{i\sqcup j}(\iota^*_{i\sqcup j}\omega), \widehat{\textbf{F}}_{i\sqcup j}(\iota^*_{i\sqcup j}\omega')\Big]\\
&=\alpha_{i\sqcup j}\Big(i\int_{M_i\sqcup M_j}G\delta\iota_{i\sqcup j}^*\omega\wedge*\delta\iota_{i\sqcup j}^*\omega'\enspace1_{\mathfrak{F}(M_i\sqcup M_j)}\Big)\\
&=0
\end{align*}
because $\omega$ and $\omega'$ are spacelike separated and, therefore, so are $\iota_{i\sqcup j}^*\omega$ and $\iota_{i\sqcup j}^*\omega'$. Now, let $\omega,\omega'\in\Omega^2_0(M,\mathbb{C})$ be spacelike separated without further restrictions or assumptions. Since the topology of $(M,g)$ coincides with the Alexandrov one and since $\omega$ and $\omega'$ being spacelike separated implies $\supp\omega\cap\supp\omega'=\emptyset$ in particular, we can cover $\supp\omega$ with diamonds $D_x$, $\supp\omega\subseteq\bigcup_{x\in\supp\omega}D_x$, and $\supp\omega'$ with diamonds $D_y$, $\supp\omega'\subseteq\bigcup_{y\in\supp\omega'}D_y$ in such a way that $D_x\cap D_y=\emptyset$ for all pairs $(x,y)\in\supp\omega\times\supp\omega'$ (smooth manifolds are $T_4$). These diamonds will turn into contractible 4-dimensional oriented globally hyperbolic spacetimes if we endow them with the structures induced by $(M,g)$. Let $(\psi^x)_{x\in\supp\omega}$ be a partition of unity subordinated to $(D_x)_{x\in\supp\omega}$ and $(\varphi^y)_{y\in\supp\omega'}$ be a partition of unity subordinated to $(D_y)_{y\in\supp\omega'}$ respectively. On account of the previous results, it holds
\begin{align*}
\Big[\widehat{\textbf{F}}_u(\omega), \widehat{\textbf{F}}_u(\omega')\Big]&=\Big[\widehat{\textbf{F}}_u(\underset{x}{\sum}\psi^x\omega), \widehat{\textbf{F}}_u(\underset{y}{\sum}\varphi^y\omega')\Big]=\Big[\widehat{\textbf{F}}_u(\underset{x}{\sum}\omega^x), \widehat{\textbf{F}}_u(\underset{y}{\sum}\omega'^y)\Big]\\
&=\underset{x,y}{\sum}\Big[\widehat{\textbf{F}}_u(\omega^x), \widehat{\textbf{F}}_u(\omega'^y)\Big]\\
&=0,
\end{align*}
since $\omega^x,\omega'^y\in\Omega^2_0(M,\mathbb{C})$ are pairwise spacelike separated and $\supp\omega^x\subseteq D_x$ and $\supp\omega'^y\subseteq D_y$. Note that the sums are finite because the partition of unity is locally finite.
\end{proof}

The proof of $\mathfrak{F}_u(M)$ obeying the principle of locality makes it clear why we wanted to consider disconnected spacetimes with finitely many connected contractible components in the first place. What remains to be shown is our claim that the commutator in $\mathfrak{F}_u(M)$ is given by the Lichnerowicz commutator:

\begin{prop}\label{propagator}
The commutator between two algebra elements in \(\mathfrak{F}_u(M)\) is provided by the so-called Lichnerowicz commutator \cite{Lich}
\begin{align}\label{Lichnerowicz commutator global field algebra}
\Big[\widehat{\bf F}(\omega), \widehat{\bf F}(\omega')\Big]=i\,(\int_MG\delta\omega\wedge*\delta\omega')\,1_{\mathfrak{F}_u(M)}\enspace\forall\omega,\omega'\in\Omega^2_0(M, \mathbb{C}),
\end{align}
where $1_{\mathfrak{F}_u(M)}$ is the identity element of the universal algebra and $G$ the causal propagator of the $\Box$-operator.
\end{prop}
\begin{proof}
We sketch here the main steps of the proof pointing an interested reader to \cite{Lang} for the details of some lengthy albeit straightforward computations. Choose a Cauchy surface \(\Sigma_Z\) in the future of \(\supp(\omega)\) and \(\supp(\omega')\) and consider the compact set \(K:=\Big(J^+\big(\supp(\omega)\big)\cap\Sigma_Z\Big)\cup\Big(J^+\big(\supp(\omega')\big)\cap\Sigma_Z\Big)\). Cover \(K\) with finitely many contractible open subsets \(U_i\), $i=1,...,n<\infty$ of \(\Sigma_Z\) whose Cauchy developments will be called \(D^M(U_i)\). Without loss of generality, all \(D^M(U_i)\) belong to the chosen cover of \(M\). Let \(V_k\) be a finite refinement of \(U_i\) such that
\begin{align}
\begin{split}\label{refinement}
&\exists i\in\mathbb{N}\enspace\text{such that}\enspace V_k\subset U_i,\\
&V_k\cap V_{k'}\not=\emptyset\Longrightarrow\exists i'\enspace\text{such that}\enspace U_{i'}\supset V_k\cup V_k'.
\end{split}
\end{align}
Such a refinement exists, because, whenever (\ref{refinement}) is not fulfilled by two sets \(V_k\) and \(V_l\), we can replace them with finitely many other sets satisfying such condition and all other constraints of our construction. 
\begin{figure}[ht!]
\centering
\includegraphics[height=8cm, width=10cm]{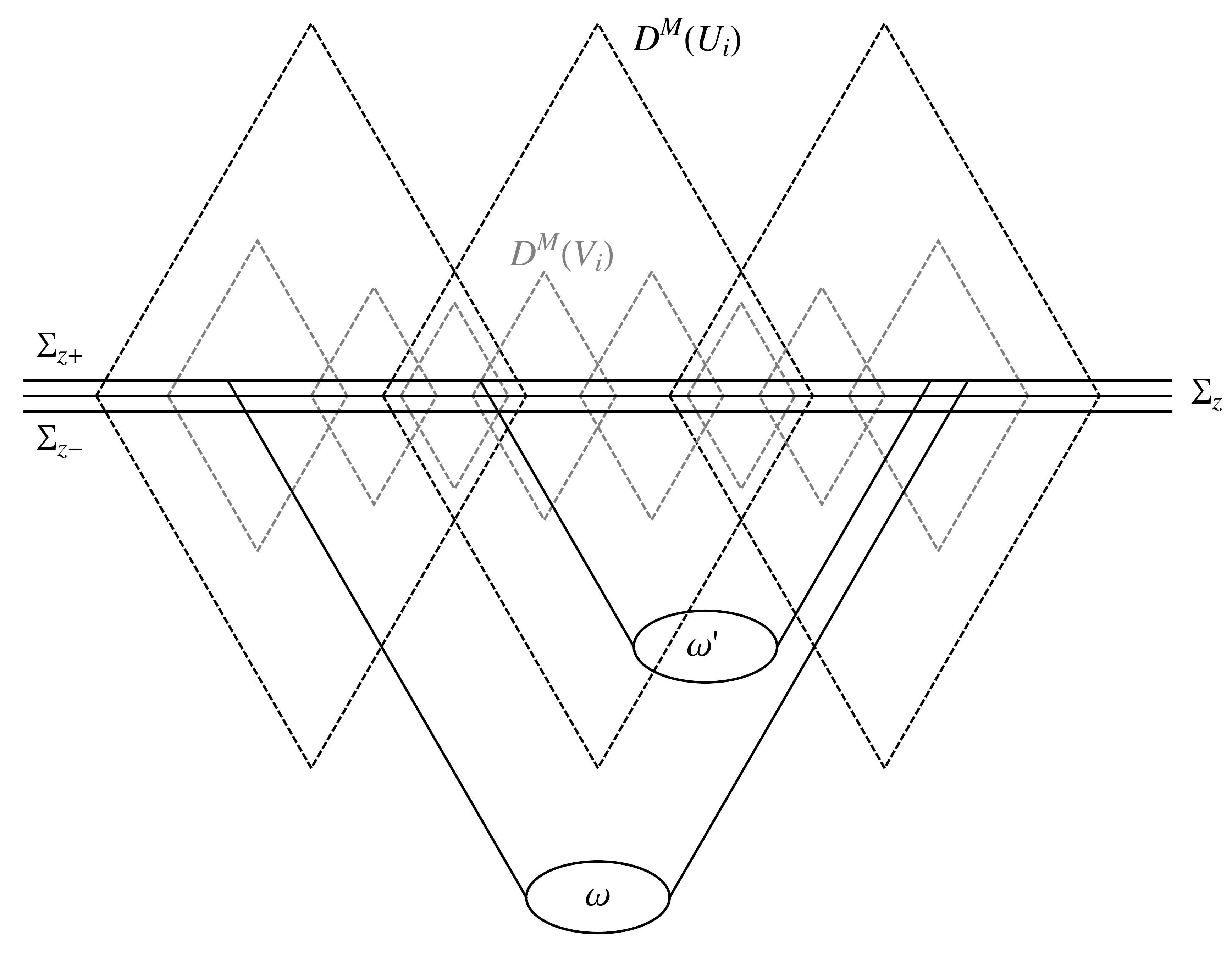}
\caption*{\scriptsize{Schematic description of the geometric loci employed in the proof of proposition \ref{propagator}.}}
\end{figure} 
\noindent Construct the Cauchy developments \(D^M(V_k)\) of each set $V_k$ in the refinement; automatically it holds that there exists $i$, such that \(D^M(V_k)\subset D^M(U_i)\). If \(D^M(V_k)\cap D^M(V_{k'})=\emptyset\) then \(D^M(V_k)\) and \(D^M(V_{k'})\) are spacelike separated. If instead \(D^M(V_k)\cap D^M(V_{k'})\not=\emptyset\) then there exists an \(i\) such that \(D^M(U_i)\supset D^M(V_k)\cup D^M(V_{k'})\). Since there exist finitely many sets \(D^M(V_k)\), we can always find a spacelike Cauchy surface \(\Sigma_{Z+}\) in the future of \(\Sigma_Z\) and a spacelike Cauchy surface \(\Sigma_{Z-}\) in the past such that \(J^+\big(\supp(\omega)\big)\cap\Sigma_{Z+}\), \(J^+\big(\supp(\omega')\big)\cap\Sigma_{Z+}\), \(J^M_+\big(\supp(\omega)\big)\cap\Sigma_{Z-}\) and \(J^M_+\big(\supp(\omega')\big)\cap\Sigma_{Z-}\) are all contained in \(\bigcup_k D^M(V_k)\).
Now, let $\chi^+,\chi^-\in C^\infty(M)$ be chosen in such a way that $\chi^++\chi^-=1$ and $\chi^+$ is identically $1$ in  \(J^+(\Sigma_{Z+})\) and $0$ in $J^-(\Sigma_{Z-})$. Consider \(\tilde{\omega}:=\omega-\Box\chi^-G_+\omega\) and \(\tilde{\omega}':=\omega'-\Box\chi^-G_+\omega'\). On account of the construction and of the properties of \(\chi^-\), \(G_+\), \(\tilde{\omega}\) and \(\tilde{\omega}'\) are compactly supported and their supports lie in \(\bigcup_kD^M(V_k)\). Choose a partition of unity \(\{\psi_k\}\) belonging to \(\{D^M(V_k)\}\). On account of the properties  EOM 1) and EOM 2), $\widehat{\bf F}(\omega)=\widehat{\bf F}(\tilde\omega)$ and $\widehat{\bf F}(\omega')=\widehat{\bf F}(\tilde\omega')$. Hence
$$\Big[\widehat{\bf{F}}(\tilde{\omega}), \widehat{\bf{F}}(\tilde{\omega}')\Big]=\sum_{k,k'}\Big[\widehat{\bf{F}}(\psi_k\tilde{\omega}), \widehat{\bf{F}}(\psi_{k'}\tilde{\omega}')\Big]=\sum_{k\sim k'}\Big[\widehat{\bf{F}}(\psi_k\tilde{\omega}), \widehat{\bf{F}}(\psi_{k'}\tilde{\omega}')\Big],$$
where $k\sim k'$ means that we consider only the pairs $(k,k')$ such that $D^M(V_k)\cap D^M(V_{k'})\neq\emptyset$ and where in the last equality we used that spacelike separated observables do commute. Hence 
\begin{align*}
\sum_{k\sim k'}\Big[\widehat{\bf{F}}(\psi_k\tilde{\omega}), \widehat{\bf{F}}(\psi_{k'}\tilde{\omega}')\Big]
= \sum_{k\sim k'}i\,\left(\int_MG\delta(\psi_k\tilde{\omega})\wedge*\delta(\psi_{k'}\tilde{\omega}')\right)\,1_{\mathfrak{F}_u(M)}=\\
=\sum_{k,k'}i\,\left(\int_MG\delta(\psi_k\tilde{\omega})\wedge*\delta(\psi_{k'}\tilde{\omega}')\right)\,1_{\mathfrak{F}_u(M)}
=i\,\left(\int_MG\delta\tilde{\omega}\wedge*\delta\tilde{\omega}'\right)\,1_{\mathfrak{F}_u(M)}
\end{align*}
$$=i\,\left(\int_MG\delta\omega\wedge*\delta\omega'\right)\,1_{\mathfrak{F}_u(M)},$$
where, in the second equality we consider all possible values for $k$ and $k'$ since the additional ones contribute $0$ to the integral. Notice that in the various identities we used the fact that all test functions are compactly supported, that $G$ commutes with both $d$ and $\delta$ and that all sums are over a finite set of indices.
\end{proof}

\subsection{The time slice axiom}
We have shown that $\mathfrak{F}_u(M)$ enjoys all the properties wanted by a genuine algebra of observables; hence we can start investigating its additional features. To start with,
\begin{lem}\label{time slice axiom}
The universal algebra \(\mathfrak{F}_u(M)\) satisfies the time slice axiom, that is, if \(\Sigma\) is a Cauchy surface of \((M, g)\) and \(\mathcal{O}\) a globally hyperbolic subset of \(M\) containing \(\Sigma\), it holds that \(\mathfrak{F}_u(\mathcal{O})=\mathfrak{F}_u(\mathcal{M})\).
\end{lem}
\begin{proof}
Let \(\mathcal{O}(\Sigma)\) be an open neighbourhood of \(\Sigma\). It is sufficient to show that for every \(\omega\in\Omega^2_0(M, \mathbb{C})\) there exists a \(\omega'\in\Omega^2_0(M, \mathbb{C})\) with \(\supp(\omega')\subset\mathcal{O}(\Sigma)\) such that \(\widehat{\bf{F}}(\omega)=\widehat{\bf{F}}(\omega')\). Since \(\mathcal{O}(\Sigma)\) is an open neighbourhood of \(\Sigma\) and \(J^\pm\big(\supp(\omega)\big)\cap\Sigma\) is compact, there exist Cauchy surfaces \(\Sigma_f\) and \(\Sigma_p\) respectively in the future and in the past of \(\Sigma\) such that \(J^\pm\big(\supp(\omega)\big)\cap\Sigma_f\subset\mathcal{O}(\Sigma)\) and \(J^\pm\big(\supp(\omega)\big)\cap\Sigma_p\subset\mathcal{O}(\Sigma)\). Let \(\chi^+,\chi^-\) lie in $C^\infty(M)$ and let us fix them in such a way that $\chi^+ + \chi^-=1$ and that $\chi^+$ vanishes in the past of $\Sigma_p$, whereas it is equal to $1$ in the future of $\Sigma_f$. Then, if we define
\begin{align*}
\omega'=\omega-\Box\chi^+G_-\omega-\Box\chi^-G_+\omega,
\end{align*}
it holds that \(\supp(\omega')\subset\mathcal{O}(\Sigma)\) is compact due to the properties of \(\chi^\pm\) and \(G_\pm\). Furthermore, from the conditions EOM 1) and EOM 2) on $\mathfrak{F}_u(M)$, it follows that \(\widehat{\bf{F}}(\omega)=\widehat{\bf{F}}(\omega')\).
\end{proof}

\subsection{The centre of $\mathfrak{F}_u(M)$}
The aim of this subsection is to investigate a distinguishing aspect of the universal algebra, namely the appearance of new features which have no counterpart in the local algebras, above dubbed as $\mathfrak{F}_c(M_i)$. From a technical point of view, this translates in the existence of a non trivial centre in $\mathfrak{F}_u(M)$, that is there exists a non trivial subalgebra whose elements are commuting with all those of the universal algebra. Yet we want to stress that this happens only if the topology of the underlying background is rather peculiar, namely if $H^2(M)\neq\{0\}$. If, on the contrary, the second de Rham cohomology group is trivial, then the equation $dF=0$ in \eqref{Maxwell} entails the existence of a global one-form $A$ such that $F=dA$. In this case the field algebra of the field strength tensor could be globally defined as the differential of that of the vector potential and no non-trivial centre would appear.

Therefore, we will henceforth assume that $H^2(M)\neq\{0\}$ and, with the next lemma, we show how to characterize the elements of the centre of \(\mathfrak{F}_u(M)\).
\begin{prop}\label{centre} An algebra element \(\widehat{\bf{F}}(\omega)\) lies in the centre of \(\mathfrak{F}_u(M)\) if and only if \(\omega=\alpha+\beta\) with \(\alpha\in\Omega^2_{0,\delta}(M, \mathbb{C})\) and \(\beta\in\Omega^2_{0,d}(M, \mathbb{C})\).
\end{prop}
\begin{proof}
\(\widehat{\bf{F}}(\omega)\) is in the centre of \(\mathfrak{F}_u(M)\) if and only if \(\Big[\widehat{\bf{F}}(\omega), \widehat{\bf{F}}(\omega')\Big]=\int_M\delta G\omega\wedge*\delta\omega'=\int_Md\delta G\omega\wedge*\omega'=0\) for all \(\omega'\in\Omega^2_0(M, \mathbb{C})\). Since $\Omega^2_0(M,\mathbb{C})$ comes endowed with the non-degenerate scalar product \(\big<\omega,\omega'\big>=\int_M\omega\wedge*\omega'\),  then the commutator between $\widehat{\bf{F}}(\omega)$ and $\widehat{\bf{F}}(\omega')$ vanishes if and only if $Gd\delta\omega=0=G\delta d\omega$. In turn, this last equality holds if and only if \(\delta d\omega=\Box\alpha\) and \(d\delta\omega=\Box\beta\), \(\alpha,\beta\in\Omega^2_0(M, \mathbb{C})\). We can exploit the properties of the Green's functions to conclude that the following chains of identities hold $0=G_\pm\delta\delta d\omega=G_\pm\Box\delta\alpha=\delta\alpha$ and equivalently $0=G_\pm dd\delta\omega=G_\pm\Box d\beta=d\beta$. Furthermore, it holds true that \(\omega=G_\pm\Box\omega=G_\pm(-\delta d\omega-d\delta\omega)=-G_\pm\Box(\alpha+\beta)=-\alpha-\beta\).
\end{proof}
Notice that the proposition guarantees that the centre is trivial if and only if \(H^2(M,\mathbb{C})=\{0\}\) since, in this case, the closedness of $\alpha$ and the coclosedness of $\beta$ would guarantee the existence of $\eta\in\Omega^3_0(M,\mathbb{C})$ and of $\theta\in\Omega^1_0(M,\mathbb{C})$ such that $\omega=d\theta+\delta\eta$. Under this assumption, on account of ${\rm EOM\, 1)}$ and of ${\rm EOM\, 2)}$ for $\mathfrak{F}_u(M)$, the field strength operator vanishes. In order to better understand this feature, it is worth to construct explicitly non trivial elements of the centre whenever \(0<\dim(H^2(M,\mathbb{C}))<\infty\), the latter bound being assumed only for the sake of simplicity. Notice that in the forthcoming analysis we will work with real forms, thus dropping the reference to $\mathbb{C}$; this does not clash with the previous results and it is assumed still only for the sake of simplicity. Out of the non-degenerateness of the scalar product on $H^2(M)$, $M$ being four dimensional, and out of Poincar\'{e} duality, \cite[Chap. 1]{BoTu}, the following chain of isomorphisms holds true:
\begin{align*}
(H^2(M))^*\cong H^2(M)\cong(H^2_{c}(M))^*\cong H^2_{c}(M),
\end{align*}
where the subscript $c$  here stands for compact support. Notice that, in the first and in the third isomorphism, the hypothesis of $H^2(M)$ being finite dimensional plays a key role. To wit both $H^2(M)$ and $H^2_c(M)$ are finite-dimensional vector spaces and hence isomorphic to their dual. Consequently every element $\lambda$ of \((H^2(M))^*\) can be represented as  
\begin{align*}
H^2(M)\ni [F]\longmapsto \lambda([F])=\int_M F\wedge\eta.
\end{align*}
Notice that the symbol $[F]$ to indicate an equivalence class in $H^2(M)$ has been chosen for a notational reason which will be manifest in the forthcoming discussion. Furthermore, on the right hand side, $F$ is an arbitrary representative of $[F]$ as well as $\eta$ is an arbitrary representative of a unique equivalence class \([\eta]\in H^2_c(M)\). By direct inspection, one can realize that the integral does not depend on the various choices. Since every \([z]\in H^\infty_2(M)\) defines a linear map $\int_z:H^2(M)\longrightarrow\mathbb{R}$, there exists a unique \([\omega_z]\in H^2_c(M)\) such that  
\begin{align*}
\int_z\omega=\int_M F\wedge\omega_z,\enspace\forall [F]\in H^2(M) 
\end{align*}
where all formulas are independent from the choice of a representative in the various equivalence classes. We can interpret the above remarks as follows: On account of the hypothesis \(H^2(M)\neq\{0\}\), there exists at least an equivalence class of non-exact field strength tensor $[F]$. As a result of that, there exists \([z]\in H^\infty_2(M)\) and \([\omega_z]\in H^2_c(M)\) fulfilling regardless of the chosen representative
\begin{align*}
\int_zF=\int_MF\wedge\omega_z\neq 0.
\end{align*}
Hence we have constructed a classical field strength $F$ whose associated algebra element \(\widehat{\bf{F}}(\omega_z)\) is a non-trivial element of the centre which can be interpreted as the magnetic flux through the 2-cycle \(z\). The very same discussion holds true also for \(*\omega_z\) in place of \(\omega_z\) because of \(\int_M F\wedge*\omega_z=\int_M *F\wedge\omega_z=\int_z*F\) for all \([F]\in H^2(M)\). From a physical point of view \(\widehat{\bf{F}}(*\omega_z)\) can be interpreted as the electric flux through \(z\). We would like to draw the attention to the fact that these non-trivial elements of the algebra give rise to superselection sectors as discussed in \cite{Ash80}.
\subsection{Maxwell field as a local covariant quantum field theory}
As the very last point of our investigation on the algebra of observables for the free Maxwell field, we address the question whether it defines a local covariant quantum field theory as per definition 2.1 in \cite{BFV}. In this section we shall use both the terminology and the nomenclature of this last cited paper; we refer to it for an extensive analysis and here we recollect instead just the definition of the main ingredients we need:
\begin{itemize}
\item $\mathfrak{GlobHyp}$: the category whose objects are $(M,g)$, that is four dimensional oriented and time oriented globally hyperbolic spacetimes with at most finitely many connected components, endowed with a smooth metric of signature $(+,-,-,-)$. A morphism between two objects $(M,g)$ and $(M',g')$ is a smooth embedding $\mu:M\to M'$ such that $\mu(M)$ is causally convex\footnote{We recall that an open subset $\mathcal{O}$ of a globally hyperbolic spacetime is called {\em causally convex} if $\forall x,y\in\mathcal{O}$ all causal curves connecting $x$ to $y$ lie entirely inside $\mathcal{O}$.}, preserves orientation and time orientation and $\mu^*g'=g$ on $M$.
\item $\mathfrak{GlobHyp}_2$: the subcategory of $\mathfrak{GlobHyp}$ whose objects are those $(M,g)\in{\rm Obj}(\mathfrak{GlobHyp})$ and $H^2(M)=\{0\}$. A morphism between two objects $(M,g)$ and $(M',g')$ is a smooth embedding $\mu:M\to M'$ such that $\mu(M)$ is causally convex, preserves orientation and time orientation and $\mu^*g'=g$ on $M$. Notice that, since $\mu(M)$ is diffeomorphic to $M$, its cohomology groups are isomorphic to those of $M$ -- \cite[Corol. 11.3]{Lee}.
\item $\mathfrak{Alg}$: the category whose objects are unital $*$-algebras whereas morphisms are injective unit-preserving $*$-homomorphisms.
\end{itemize}
 Since the composition map between morphisms and the existence of an identity map are straightforwardly defined in every case we shall consider, we will omit them. We shall start proving a weaker form of general local covariance, where the class of spacetimes we consider is not the most general one. We wish to postpone the explanation for this choice to after the proof of the following proposition since we feel that reading it will make our point clearer than an abstract a priori argument.
 
\begin{prop}
There exists a covariant function ${\rm F}_u:\mathfrak{GlobHyp}_2\longrightarrow\mathfrak{Alg}$ which assigns to every object $(M,g)$ in $\mathfrak{GlobHyp}_2$ the $*$-algebra $\mathfrak{F}_u(M)$ with the induced action on the morphisms. 
In diagrammatic form:
\begin{equation*}
\begin{CD}
 (M,g)@>{\mu}>> (M',g') \\
 @V{{\rm F}_{u}}VV        @VV{{\rm F}_{u}}V\\
\mathfrak{F}_u(M) @>{\alpha_\mu}>> \mathfrak{F}_u(M')
\end{CD}
\end{equation*}
Here $\alpha_\mu$ is the unit-preserving $*$-homomorphism defined by its action on the generators as $\alpha_\mu\left(\widehat{\bf F}(\omega)\right):= \widehat{\bf F}(\mu_*\omega)$ where $\mu_*\omega$ is the pull-back of $\omega$ via $\mu^{-1}:\mu(M)\to M$.
Furthermore, such local covariant quantum field theory is causal and it fulfils the time slice axiom.
\end{prop}
\begin{proof}
As discussed at the beginning of the session, we can associate to each $(M,g)\in{\rm Obj}(\mathfrak{GlobHyp})$ the universal algebra along the lines of the previous section. Hence, if we consider any morphism $\mu$ between two objects $(M,g)$ and $(M',g')$, we can consider $(\mu(M),g'|_{\mu(M)})$ as a globally hyperbolic spacetime on its own. Since $\mu$ is an isometry, it means that any covering of $M$ via globally hyperbolic contractible subsets $M_i$, $i=1,...,n<\infty$ induces a cover of $\mu(M)$ via $\mu(M_i)$. It is easy to realize that $\mathfrak{F}_c(\mu(M_i))=\alpha_\mu\left(\mathfrak{F}_c(M_i)\right)$ where $\alpha_\mu$ acts on each generator $\widehat{\bf F}(\omega)$, $\omega\in\Omega^2_0(M_i)$ yielding $\widehat{\bf F}(\mu_*\omega)$. Notice that, since $d$ is independent from the metric and $\delta$ is constructed out of $d$ and of the Hodge dual $*$, they both commute with isometries. Hence $\mu_*d\widehat{\bf F}(\omega)=d\mu_*\widehat{\bf F}(\omega)$ and $\mu_*\delta\widehat{\bf F}(\omega)=\delta\mu_*\widehat{\bf F}(\omega)$. This also suffices to claim that, if we call $G_\mu$ the causal propagator of the $\Box$-operator on $\mu(M)$, it holds that $\mu_*\circ G=G_\mu\circ \mu_*$. Hence we can consider the commutator between two generators to prove
\begin{gather*}
\left[\alpha_\mu(\widehat{\bf F}(\omega)),\alpha_\mu(\widehat{\bf F}(\omega'))\right]=i\int_{\mu(M)}\!\!G_\mu(\mu_*\omega)\wedge *d\delta(\mu_*\omega')=\\
i\int_{\mu(M)}\!\!\mu_*(G\omega)\wedge \mu_*(*d\delta\omega')=\int_{\mu(M)}\!\!\mu_*(G\omega\wedge *d\delta\omega')=\\
=i\int_M G\omega\wedge *d\delta\omega'=\left[\widehat{\bf F}(\omega),\widehat{\bf F}(\omega')\right].
\end{gather*}
Since complex conjugation is not affected by isometric embeddings, we have proven that $\mu_*$ actually defines a unit preserving $*$-homomorphism between $\mathfrak{F}_c(M_i)$ and $\mathfrak{F}_c(\mu(M_i))$. We can now without loss of generality assume that the collection of $\mu(M_i)$ is part of a covering of $M'$ with globally hyperbolic contractible spacetimes. On account of the structural properties of the universal algebra and of the absence of a centre in both $\mathfrak{F}_u(M)$ and $\mathfrak{F}_u(M')$  this entails that $\alpha_\mu$ is indeed an injective $*$-homomorphism. Furthermore on account of the commutator being defined out of the causal propagator, the theory is causal and the time-slice axiom is fulfilled as already proven in lemma \ref{time slice axiom}.
\end{proof}

We need to answer why one is forced to restrict the attention to backgrounds with trivial second de Rham cohomology group. As one can realize from the above proof, if we would have considered $\mathfrak{GlobHyp}$, one would have to consider the homomorphism induced by the embedding $\mu$ from $M$ into $M'$. Since $M$ is diffeomorphic to $\mu(M)$, it is known that these two spacetimes have isomorphic cohomology groups, but we have to go one step further and see $\mu(M)$ as an open subset of $M'$. Here is the source of potential problems since, even if $H^2(M)\neq\{0\}$, there is no reason why $H^2(M')$ should be isomorphic to $H^2(M)$; actually it can also be trivial. 

We provide an explicit example: Let us consider the ultrastatic globally hyperbolic spacetime $M=\mathbb{R}\times(\frac{\pi}{4},\frac{3\pi}{4})\times\mathbb{S}^2$ endowed with the line element $ds^2=dt^2-d\chi^2-\sin^2\chi d\mathbb{S}^2(\theta,\varphi)$ where $d\mathbb{S}^2(\theta,\varphi)$ is the canonical metric of the unit $2$-sphere. By K\"unneth formula -- \cite[Chap. 1, \S 5]{BoTu}, $H^2(M)=\oplus_{p+q=2} H^p(\mathbb{R}\times (\frac{\pi}{4},\frac{3\pi}{4}))\otimes H^q(\mathbb{S}^2)$ which is non trivial since $H^2(\mathbb{S}^2)=\mathbb{R}$. Let us now consider as $M'$, the ultrastatic spacetime $\mathbb{R}\times\mathbb{S}^3$ whose metric coincides in a local chart to $ds^2$. It is manifest that $M$ is isometrically embedded in $M'$, but still K\"unneth formula entails that 
$H^2(\mathbb{R}\times\mathbb{S}^3)=\oplus_{p+q=2}H^p(\mathbb{R})\times H^q(\mathbb{S}^3)$. Since $\mathbb{R}$ is contractible, only $q=2$ contributes and therefore the second cohomology group of $\mathbb{R}\times\mathbb{S}^3$ is trivial.

Let us now consider in the framework outlined above $\omega\in\Omega^2_{0,\delta}(M)$, then $\widehat{\bf F}(\omega)$ lies in the centre of $\mathfrak{F}_u(M)$ thanks to proposition \ref{centre}. Under the isometric embedding $\mu:M\hookrightarrow M'$, one obtains $\alpha_\mu\left(\widehat{\bf F}(\omega)\right)=\widehat{\bf F}(\mu_*\omega)$. Yet, since $\mu_*$ commutes with $\delta$, $\mu_*\omega$ is coclosed and since $H^2(M')$ is trivial, there exists $\lambda\in\Omega^3_0(M')$ such that $\mu_*\omega=\delta\lambda$. This entails that $\widehat{\bf F}(\mu_*\omega)=\widehat{\bf F}(\delta\lambda)=d\widehat{\bf F}(\lambda)=0$ on account of Maxwell's equation. Barring a minor generalization, this entails that every element of the centre of $\mathfrak{F}_u(M)$ is mapped into (the equivalence class of) $0$ in $\mathfrak{F}_u(M')$. This is tantamount to claim that $\alpha_\mu$ cannot be an injective $*$-homomorphism, injectivity failing to be achieved. 

\section{Conclusions}

In this paper we have developed a full-fledged quantization scheme for the field strength tensor obeying Maxwell's equations. Since we wanted to keep the discussion as general as possible we have neither used the vector potential as an auxiliary tool nor we have assumed the compactness of the Cauchy surface of the underlying globally hyperbolic spacetime $M$. This forced us to use two-forms $F$ obeying \eqref{Maxwell} as the building block of the theory; we have shown in particular that it still possible to construct a field algebra whose generators obey the commutation relations provided by the Lichnerowicz propagator. 
Yet we have also proven that the overall procedure does not fit in the scheme of general local covariance as developed in \cite{BFV} since there exist spacetimes $M$ with $H^2(M,\mathbb{C})\neq\{0\}$. In this case the universal algebra $\mathfrak{F}_u(M)$ possesses a non trivial centre whose elements have been fully characterized in proposition \ref{centre}. Nonetheless it is possible to conceive that $M$ is isometrically embedded in a second globally hyperbolic spacetime $M'$ which has a trivial second de Rham cohomology group and thus the associated field algebra has a trivial centre. This translates in the failure of the homomorphism from $\mathfrak{F}_u(M)$ into $\mathfrak{F}_u(M')$ from being injective and thus the embedding translates in a loss of a qualitative feature of the field algebra of $M$ when seen from $M'$, such as the presence of superselection sectors as first discussed in \cite{Ash80}. 

As we have proven in the previous section, a potential way out is to restrict the class of spacetimes we consider and general local covariance is restored as soon as we assume to work only with backgrounds with vanishing second de Rham cohomology group. Yet it is fair to admit that the situation is rather puzzling: On the one hand the proposed solution would discard spacetimes, such as Schwarzschild, which are certainly of physical relevance, while on the other hand the requirement that $H^2(M,\mathbb{C})=\{0\}$ vanishes entails that all field strength tensors would descend from a vector potential. This feature is certainly desirable as soon as we want to move from a free field theory to an interacting one such as quantum electrodynamics where the spinor fields are known to interact via $A\in\Omega^1(M)$ rather than via the Faraday tensor. 

Yet we feel it is still early to claim we have a total loss: As a matter of fact, if we focus on any equivalence class $[F]\in H^2(M)$, we are considering all elements of the form $F+dA$ where $A\in\Omega^1(M)$ while $F\in\Omega^2(M)$. In other words each non trivial cohomology class is composed of two parts. The first, is responsible for qualitative features such as global topological charges or, from the quantum perspective, for the identification of a specific superselection section and, hence, it is strictly tied to the specific chosen spacetime. The second is instead tied to a $1$-form, a sort of vector potential, and it is well-suited both to discuss interactions and to apply the principle of general local covariance. Although we are aware that this is simply a remark which does not necessarily solve all problems we have at hand, we still feel it is a starting point for further investigations which is worth to consider in detail. 

\section*{Acknowledgements}
The work of C.D. is supported partly by the University of Pavia and by the Indam-GNFM project {``Stati quantistici di Hadamard e radiazione di Hawking da buchi neri rotanti''}. We gratefully acknowledge the kind hospitality of the II. Institute f\"ur Theoretische Physik of the University of Hamburg during the realization of part of this work. We are also indebted for enlightening discussions and comments with Klaus Fredenhagen, Thomas-Paul Hack, Valter Moretti, Nicola Pinamonti, Jan Schlemmer as well as with the whole LQP research group in Hamburg. We are grateful to Andrey Saveliev for drawing the picture present in the main body of the text. The content and result of this paper are partly inspired to those present in the diploma thesis of B.L., \cite{Lang}. The details of the construction of the universal algebra were worked out in York and we would like to thank Chris Fewster for numerous discussions, in particular for his suggestion to consider disconnected spacetimes as well in order to realise that the universal algebra obeys the principle of locality.



\begin{thebibliography}{999}

\bibitem[AS80]{Ash80}
  A.~Ashtekar, A.~Sen,
  {\em ``On The Role Of Space-time Topology In Quantum Phenomena: Superselection Of Charge And Emergence Of Nontrivial Vacua,''}
  J.\ Math.\ Phys.\  {\bf 21 } (1980)  526.

\bibitem[BGP07]{BGP}
C.~B\"ar, N.~Ginoux and F.~Pf\"affle,
\emph{``Wave Equations on Lorentzian Manifolds and Quantization''}
(2007) European Mathematical Society.

\bibitem[BS03]{Bernal}
  A.~N.~Bernal, M.~Sanchez,
  {\em ``On Smooth Cauchy hypersurfaces and Geroch's splitting theorem,''}
  Commun.\ Math.\ Phys.\  {\bf 243 } (2003)  461-470,
  [gr-qc/0306108].
  
\bibitem[BS05]{Bernal3}
  A.~N.~Bernal and M.~Sanchez,
  {\em ``Smoothness of time functions and the metric splitting of globally hyperbolic space-times,''}
  Commun.\ Math.\ Phys.\  {\bf 257} (2005) 43
  [gr-qc/0401112].  

\bibitem[BS06]{Bernal2}
  A.~N.~Bernal, M.~Sanchez,
  {\em ``Further results on the smoothability of Cauchy hypersurfaces and Cauchy time functions,''}
  Lett.\ Math.\ Phys.\  {\bf 77 } (2006)  183-197,
  [gr-qc/0512095].
  
\bibitem[BFV03]{BFV}
  R.~Brunetti, K.~Fredenhagen and R.~Verch,
{\em ``The generally covariant locality principle: A new paradigm for local quantum physics,''}
  Commun.\ Math.\ Phys.\  {\bf 237} (2003) 31, 
  [arXiv:math-ph/0112041].

\bibitem[Bon77]{Bongaarts}
  P.~J.~M.~Bongaarts,
  {\em ``Maxwell's Equations in Axiomatic Quantum Field Theory. 1. Field Tensor and Potentials,''}
  J.\ Math.\ Phys.\  {\bf 18 } (1977)  1510.
  
\bibitem[BT95]{BoTu}
R.~Bott, L.~W.~Tu 
{\em ``Differential Forms in Algebraic Topology''} (1995) Springer-Verlag.

\bibitem[Dap11]{Dappiaggi}
  C.~Dappiaggi,
  {\em ``Remarks on the Reeh-Schlieder property for higher spin free fields on curved spacetimes,''}  
  [arXiv:1102.5270 [math-ph]].

\bibitem[Dim92]{Dimock}
  J.~Dimock,
  {\em ``Quantized electromagnetic field on a manifold,''}
  Rev.\ Math.\ Phys.\  {\bf 4}, 223 (1992).
  
\bibitem[FP03]{Fewster}
  C.~J.~Fewster, M.~J.~Pfenning,
  {\em ``A Quantum weak energy inequality for spin one fields in curved space-time,''}
  J.\ Math.\ Phys.\  {\bf 44 }, (2003)  4480,
  [gr-qc/0303106].

\bibitem[Fre89]{Fredenhagen}
  K.~Fredenhagen,
  {\em ``Generalizations of the theory of superselection sectors,''}
  In {\it Palermo 1989, Proceedings, The algebraic theory of superselection sectors and field theory} (1989), 379-387.

\bibitem[Fre95]{Fredenhagen2}
  K.~Fredenhagen,
{\em ``Superselection Sectors''} available at \url{http://unith.desy.de/sites/site_unith/content/e20/e72/e180/e193/infoboxContent203/superselect.ps.gz}.

\bibitem[Hol08]{Hollands}
  S.~Hollands,
  {\em ``Renormalized Quantum Yang-Mills Fields in Curved Spacetime,''}
  Rev.\ Math.\ Phys.\  {\bf 20 } (2008)  1033,
  [arXiv:0705.3340 [gr-qc]].

\bibitem[Kus10]{Kusku}
  M.~K\"usk\"u,
  {``The free Maxwell field in curved space-time,''}
  Diplomarbeit (2001) - Universit\"at Hamburg, available at \url{ftp://ftp.desy.de/pub/preprints/desy/thesis/desy-thesis-01-040.ps.gz}.
  
\bibitem[Lan10]{Lang}
B.~Lang,
{\em ``Homologie und die Feldalgebra des quantisierten Maxwellfeldes''},
Diplomarbeit -  (2010) Universit\"at Freiburg, available at \url{http://www.desy.de/uni-th/theses/Dipl_Lang.pdf}.  

\bibitem[Lee03]{Lee}
J.~M.~Lee, 
{\em ``Introduction to smooth manifolds''} (2003) Springer-Verlag.

\bibitem[Lic61]{Lich}
A.~Lichnerowicz, 
{\em ``Propagateurs et commutateurs en relativit\'e gen\`erale''},
Inst. Hautes \'Etudes Sci. Pub. Math. {\bf 10} (1961) 56.

\bibitem[Pfe09]{Pfenning}
  M.~J.~Pfenning,
  {\em ``Quantization of the Maxwell field in curved spacetimes of arbitrary dimension,''}
  Class.\ Quant.\ Grav.\  {\bf 26 }, (2009) 135017, [arXiv:0902.4887 [math-ph]].

\bibitem[Sak94]{Sakurai}
J.~J.~Sakurai,
{\em ``Modern Quantum Mechanics,''}
(1994) Addison-Wesley Publishing.

\bibitem[San10]{Sanders}
K.~Sanders,
{\em ``The locally covariant Dirac field,''}
Rev. Math. Phys. {\bf 22}, (2010) 381.

\end{thebibliography}
\end{document}